\documentclass[10pt,aps,twocolumn,longbibliography,superscriptaddress,prx]{revtex4-2}
\usepackage[english]{babel}
\usepackage[utf8]{inputenc}
\usepackage[colorinlistoftodos, color=green!40, prependcaption]{todonotes}
\usepackage{amsthm}
\usepackage{amssymb}
\usepackage{comment}
\usepackage{mathtools}
\usepackage{physics}
\usepackage{chngcntr}
\usepackage{xcolor}
\usepackage{graphicx}
\usepackage[export]{adjustbox}
\usepackage{graphbox}
\usepackage{placeins}
\usepackage[T1]{fontenc}
\usepackage{lipsum}
\usepackage{centernot}
\usepackage{csquotes}
\mathtoolsset{showonlyrefs}
\usepackage[pdftex, pdftitle={Article}, pdfauthor={Author},colorlinks = true, linkcolor = blue, urlcolor  = blue, citecolor = red]{hyperref} 

\makeatletter
\@addtoreset{equation}{myequation}
\makeatother

\makeatletter
\@addtoreset{theorem}{mytheorem}
\makeatother

\makeatletter
\@addtoreset{section}{mysection}
\makeatother

\usepackage{amsfonts}
\usepackage{amsmath}
\usepackage{stmaryrd}
\usepackage{braket}
\usepackage{dsfont}
\usepackage{bbold}
\usepackage{relsize}
\usepackage{float}
\usepackage[toc,page,titletoc]{appendix}
\usepackage[font=scriptsize]{caption}
\let\ACMmaketitle=\maketitle
\renewcommand{\maketitle}{\begingroup\let\footnote=\thanks \ACMmaketitle\endgroup}

\newcommand{\tot}{\text{tot}}
\newcommand{\Q}[1]{\mathcal{Q}^{(#1)}}
\newcommand{\Pcap}[1]{\mathcal{P}^{(#1)}}
\newcommand{\Ccap}[1]{\mathcal{C}^{(#1)}}

\newtheorem{theorem}{Theorem}

\newtheorem*{definition*}{Definition}

\newtheorem{lemma}{Lemma}[section]

\newtheorem*{conjecture*}{Conjecture}

\theoremstyle{definition}

\newcommand{\PRLsep}{\noindent\makebox[\linewidth]{\resizebox{0.3333\linewidth}{1pt}{$\blacklozenge$}}\bigskip}

\begin{document}
\title{Simultaneous superadditivity of the direct and complementary channel capacities}
\author{Satvik Singh}
\author{Sergii Strelchuk}
\affiliation{DAMTP, Centre for Mathematical Sciences, University of Cambridge, Cambridge CB30WA, UK}

\begin{abstract}
Quantum communication channels differ from their classical counterparts because their capacities can be superadditive. The principle of monogamy of entanglement suggests that superadditive improvements in the transmission capacity of a channel should reduce the amount of information loss to the environment. We challenge this intuition  by demonstrating that the coherent and private information of a channel and its complement can be simultaneously superadditive for arbitrarily many channel uses. To quantify the limits of this effect,
we consider the notion of max (resp. total)  private information of a channel, which represents the maximum (resp. sum) of the private information of the channel itself and its complement, and study its relationship with the coherent information of the individual direct and complementary channels.   
For a varying number of channel uses, we show that these quantities can obey different interleaving sequences of inequalities. 
\end{abstract}

\maketitle

Quantum channels have several intriguing properties that separate them from their classical counterparts. 
One of them is the ability to send information superadditively \cite{divincenzo1998quantum}, whereby multiple uses of the same channel increase the amount of information that it can reliably transmit. Another one, namely superactivation \cite{Smith2008superactivate},  demonstrates how some channels, having initially no capacity to send information can regain it when combined with an equally useless zero-capacity channel.
Both of these effects were introduced and subsequently studied in the context of having access to a channel ${\cal N}_{A\to B}$ where Alice, the sender, communicates with the receiver, Bob. 

One can define different kinds of capacities of $\mathcal{N}$ depending on the type of information being sent. The quantum capacity ${\cal Q}(\cal N)$ of $\mathcal{N}$ quantifies the maximum rate at which Alice can send quantum information reliably to Bob and can be expressed as a regularization of the channel's coherent information: $\mathcal{Q}(\mathcal{N}) = \lim_{k\to \infty} \mathcal{Q}^{(k)}(\mathcal{N}), \mathcal{Q}^{(k)}(\mathcal{N}) := \mathcal{Q}^{(1)}(\mathcal{N}^{\otimes k})/k, \mathcal{Q}^{(1)}(\mathcal{N}) := \sup_{\rho_A} [S(B)-S(E)]$.
The optimization here is over all states $\rho_A$ and the von Neumann entropies $S(B)$ and $S(E)$ are evaluated on ${\cal N}_{A\to B}(\rho_A)$ and ${\cal N}^c_{A\to E}(\rho_A)$, respectively, where the complementary channel ${\cal N}^c_{A\to E}$ models the nature of information loss to the environment (Eve) (its precise definition will be introduced shortly). Similarly, the classical capacity ${\cal C}({\cal N})$ of $\mathcal{N}$ quantifies the maximum rate at which Alice can send classical information reliably to Bob. In addition, if the information being sent is to be kept private from Eve, one obtains the private capacity ${\cal P}({\cal N})$ of the channel. These capacities admit regularized expressions in terms of the channel's Holevo information: $\mathcal{C}(\mathcal{N}) = \lim_{k\to \infty}  \mathcal{C}^{(k)}(\mathcal{N}),
    \mathcal{C}^{(k)}(\mathcal{N}) := \mathcal{C}^{(1)}(\mathcal{N}^{\otimes k})/k,
    \mathcal{C}^{(1)}(\mathcal{N}) := \sup_{\rho_{XA}} I(X:B),$ and private information $\mathcal{P}(\mathcal{N}) = \lim_{k\to \infty}  \mathcal{P}^{(k)}(\mathcal{N}), 
    \mathcal{P}^{(k)}(\mathcal{N}) := \mathcal{P}^{(1)}(\mathcal{N}^{\otimes k})/k,
    \mathcal{P}^{(1)}(\mathcal{N}) := \sup_{\rho_{XA}} [I(X:B) - I(X:E) ].$ The optimizations above are over all classical quantum states $\rho_{XA}=\sum_{x} p_x \ketbra{x}_X \otimes \rho^x_A$ and the mutual information terms $I(X:B)$ and $I(X:E)$ are evaluated on $\mathcal{N}_{A\to B}(\rho_{XA})$ and $\mathcal{N}^c_{A\to E}(\rho_{XA})$, respectively.

Since the first demonstration~\cite{divincenzo1998quantum}, there have been a variety of results that illustrate striking superadditive behavior of quantum channel capacities~\cite{wilde2013quantum,zhu2018superadditivity,zhu2017superadditivity,cubitt2015unbounded, Elkouss2015private,elkouss2016nonconvexity,smith2009extensive,brandao2012does,shirokov2015superactivation,leditzky2022generic,leung2014maximal,li2009private}. For instance, the $k$-letter coherent and private information of a channel obey the inequality: ${\cal Q}^{(k)}({\cal N})\le {\cal P}^{(k)}({\cal N}) \,\, \forall k$. However, for different numbers of channel uses, the coherent information of $\cal N$ can exceed its private information \cite{Elkouss2015private}, making the former inequality valid only for a fixed $k$.

The above effects were all demonstrated in the setting where Alice optimizes her data transmission rate to Bob while minimizing information `leakage' to Eve, who behaves as a non-participating party during transmission. There are several results that investigate information transmission problems under non-trivial behaviour of the environment~\cite{karumanchi2016classical, oskouei2021capacities,winter2005environment}. In one of the earliest works~\cite{winter2005environment} of such kind, the authors investigated the capacities of quantum channels by allowing Eve to locally
measure and communicate classical messages to Bob. Later, this was extended to allow a helper \cite{karumanchi2016classical} -- a benevolent third party -- who can adjust the environment state. This enabled one to derive streamlined examples for super-additivity due to the extra abilities of the helper to adjust the state of the environment depending on the message being sent. For example, the so-called locking capacity~\cite{guha2014quantum} of a channel in this regime is superadditive, whereas in the absence of auxiliary resources the question is still open.

The above scenarios supplement the direct channel ${\cal N}_{A\to B}$ with extra resources which are extrinsic to its definition. This precludes one from learning about the total capacity for information transmission by using the channel alone. 
One may observe that defining a direct channel ${\cal N}_{A\to B}$ from Alice to Bob fixes the behaviour of information loss to the environment (Eve) via the complementary channel $\mathcal{N}^c_{A\to E}$. Indeed, the Stinespring dilation theorem~\cite{stinespring1955positive,paulsen2002completely} shows that there exists an isometry $V:\mathcal{H}_A \to \mathcal{H}_B\otimes \mathcal{H}_E$ such that $\mathcal{N}(X)=\operatorname{Tr}_E(VX V^\dagger)$ and $\mathcal{N}^c (X)  = \operatorname{Tr}_B(VXV^\dagger)$. 

When optimizing the communication rates, we naturally want to take full advantage of the superadditive properties of the channel. The monogamy of entanglement principle suggests that when the direct channel is superadditive, its complementary channel to the environment is likely to have constrained data transmission capabilities. Contrary to this intuition, we show that superadditivity can persist for an unbounded number of channel uses in the strongest possible sense in both the direct and complementary channels \emph{simultaneously}.

To demonstrate this surprising effect we consider two quantities: the \emph{max} and \emph{total} coherent information of $\cal N$: 
\begin{align}
    \mathcal{Q}^{(k)}_{\max}(\mathcal{N}) &:= \max\{ \mathcal{Q}^{(k)}(\mathcal{N}), \mathcal{Q}^{(k)}(\mathcal{N}^c) \}, \\
    \mathcal{Q}^{(k)}_{\tot}(\mathcal{N}) &:= \mathcal{Q}^{(k)}(\mathcal{N}) + \mathcal{Q}^{(k)}(\mathcal{N}^c).
\end{align}
The max- and total Holevo and private information quantities can be defined similarly. The \emph{max} and \emph{total} quantum capacities can be obtained by taking $k\to \infty$ limits: 
\begin{align}
    \lim_{k\to \infty}\mathcal{Q}^{(k)}_{\max}(\mathcal{N}) &= \max\{ \mathcal{Q}(\mathcal{N}), \mathcal{Q}(\mathcal{N}^c) \}, \\
    \lim_{k \to \infty}\mathcal{Q}^{(k)}_{\tot}(\mathcal{N}) &= \mathcal{Q}(\mathcal{N}) + \mathcal{Q}(\mathcal{N}^c).
\end{align}
Operationally, Bob and Eve are now placed on an equal footing and Alice simply wants to send information at the best rate possible regardless of who plays the role of the receiver. She can either use the direct channel $\mathcal{N}_{A\to B}$ or its complement $\mathcal{N}^c_{A\to E}$ to do so, thus arriving at the max rate. In such a scenario, it is crucial to analyse the superadditive behaviour of both the direct and complementary channels together to determine which one has higher capacity to transmit information. A similar quantity was introduced in \cite{Singh2022main} in the context of entanglement distillation. On the other hand, to our best knowledge, the expression for the total quantum capacity first appeared in \cite{hirche2022bounding}, where it was used to bound the difference between the quantum and private capacities of any channel $\mathcal{N}$. Intuitively, being the sum of the direct and complementary channel capacities, the total capacity quantifies the overall information transmission capability of $\mathcal{N}$.

The above setting should be distinguished from that of quantum broadcast channels~\cite{yard2011quantum,laurenza2017general} where two (or more) recipients (Bobs) share a joint environment. As such, this model is not representative of the concepts of max and total information where the notion of the environment does not feature. 

We now briefly describe our results. For $n,\alpha\in \mathbb{N}$ satisfying $n^{\alpha-2}>8$, we construct a channel ${\cal N}$ such that (Theorem~\ref{th:main1}): 
\begin{align}\label{eq:main1}
\forall k\leq n: \quad \Q{k+1}({\cal N}), \Q{k+1}({\cal N}^c) > {\cal P}^{(k)}_{\max}({\cal N}).
\end{align}
Thus, for any number $k$ of channel uses, the max $k$-letter private information of $\cal N$ can be exceeded by the coherent information of both the direct and complementary channels by using just one extra copy of each of these channels. This means that the coherent and private information quantities of $\cal N$ are curiously interleaved:
\begin{align*}
    \Q{1}({\cal N}), \Q{1}({\cal N}^c)\leq {\cal P}^{(1)}_{\max}({\cal N}) &<  \Q{2}({\cal N}), \Q{2}({\cal N}^c) \\ 
    & \leq {\cal P}^{(2)}_{\max}({\cal N}) < \ldots 
\end{align*}

Remarkably, by choosing $n$ large enough, this phenomenon can be made to persist for arbitrarily many channel uses. Moreover, the parameter $p$ can be tuned to boost the superadditivity of the direct channel relative to its complement or vice versa. More precisely, when $1/3\leq p \leq 1/2 - 1/n^{\alpha -1}$ and $n^{\alpha -2}>12$, even though both the direct and complementary channels are still superadditive, Eq.~\eqref{eq:main1} now only holds for $\cal N$ and not for ${\cal N}^c$. Thus, the superadditivity of the direct channel dominates that of its complement (Theorem~\ref{th:main2}). 
For even smaller values of $p$, this effect becomes extreme: even the total $k-$letter private information of the channel (below a certain threshold $k$ value) can be exceeded by the coherent information of the direct channel alone, provided that it is used sufficiently many times $j>c_k$ (Theorem~\ref{th:main3}):
\begin{equation}
    \Q{j}({\cal N}) > {\cal P}^{(k)}({\cal N}) + {\cal P}^{(k)} ({\cal N}^c) = \Pcap{k}_{\tot} ({\cal N}). 
\end{equation}
In all the above cases, it is possible to precisely quantify the effects of superaddivity by computing lower bounds on the difference quantities such as $\Q{k+1}({\cal N})- \Pcap{k}_{\max}({\cal N})$ and $\Q{k+1}({\cal N}^c)- \Pcap{k}({\cal N}^c)$. 

Finally, it turns out that simultaneous superadditivity of both the direct and complementary channels is a non-trivial phenomenon. We prove this by constructing a channel with superadditive quantum capacity whose complement has additive capacity.

\noindent
{\it {The main construction.--}}
Our channels ${\cal N}_{n,p,d}$ are made of two building blocks: the erasure and `rocket' channels. The $d$-dimensional erasure channel $\mathcal{E}_{p,d}:A\to B$ with erasure parameter $p\in [0,1]$ takes a $d$-dimensional input and replaces it with an erasure flag $\ketbra{e}$ (orthogonal to the input space) with probability $p$ and does nothing to it otherwise: $\mathcal{E}_{p,d}(\rho) = (1-p)\rho + p\operatorname{Tr}(\rho)\ketbra{e}$. Its complement is again of the erasure type: $\mathcal{E}_{p,d}^c = \mathcal{E}_{1-p,d}$.
The capacities of $\mathcal{E}_{p,d}$ are well known: $\mathcal{Q}(\mathcal{E}_{p,d}) = \mathcal{P}(\mathcal{E}_{p,d}) = \max\{ (1-2p)\log d,0\}, \, \mathcal{C}(\mathcal{E}_{p,d}) = (1-p)\log d.$

The $d$-dimensional rocket channel \cite{Smith2009private}  $R_d: A_1 \otimes A_2 \to B$ takes two $d$-dimensional quantum systems $(A_1 \text{ and } A_2)$ as inputs and applies local random unitaries on each input \footnote{It actually suffices to apply local unitaries from a unitary 2-design, such as the Clifford group \cite{christoph2009design}.} followed by a controlled phase coupling $P=\sum_{i,j}\omega^{ij}\ketbra{i}_{A_1}\otimes \ketbra{j}_{A_2}$, where $\omega=e^{i2\pi /d}$. Finally, $A_2$ is discarded and Bob gets $A_1$ along with classical information about which local unitaries were applied. For the complementary channel $R_d^c :A_1\otimes A_2\to E$, $A_1$ is discarded and Eve gets $A_2$ along with the same classical information about the local unitaries. Since $R_d$ dephases the input registers in a random basis unknown to Alice, it has little capacity to transmit information on its own: $\mathcal{C}(R_d) \leq 2$. The same argument applies to $R^c_d$ as well: $\mathcal{C}(R_d^c)\leq 2$ (the proof of~\cite{Smith2009private} goes through by swapping labels for Bob and Eve).
However, when Alice and Bob already share a maximally entangled state (this can be achieved with probability $1-p$ by using the erasure $\mathcal{E}_{p,d}$), it turns out that Bob can undo the random phase coupling, thus allowing Alice to send quantum information at rate $\log d$ \cite{Smith2009private}. More precisely, we have
\begin{equation}\label{eq:rocket-erasure-supadd}
    \mathcal{Q}^{(1)}(R_d \otimes \mathcal{E}_{p,d}) \geq (1-p)\log d.
\end{equation}
A slight modification of this argument can be used to obtain the same result for $R^c_d$ as well:
\begin{equation}\label{eq:rocket-erasure-comp-supadd}
    \mathcal{Q}^{(1)}(R^c_d \otimes \mathcal{E}_{p,d}) \geq (1-p)\log d.
\end{equation}

An intuitive graphical proof of the above two claims (Eqs.~\eqref{eq:rocket-erasure-supadd} and \eqref{eq:rocket-erasure-comp-supadd}) is provided in Appendix~\ref{appen:rocket}.

We are now ready to introduce our main channels. For $n,d\in \mathbb{N}$ and $p\in [0,1]$, we define 
\begin{align}
    \mathcal{N}_{n,p,d} &:= R^{\otimes n}_{d} \oplus \mathcal{E}_{p,d}, \\
    \mathcal{N}^c_{n,p,d} &= R^{ c \,\otimes n}_{d} \oplus \mathcal{E}_{1-p,d}. 
\end{align}
The direct sum construction \cite{Fukuda2007additivity} allows Alice to control which of the two channels is being applied at the outset, so that the coherent and private information of such channels is just the maximum of its building blocks, see Lemma~\ref{lemma:directsum}. In our case, by using the channel $k+1$ times, Alice gains access to blocks of the form 
\begin{equation}
    R_d^{\otimes n} \otimes {\cal E}_{p,d}^{\otimes k} \,\,\,\text{or} \,\, R_d^{c\, \otimes n} \otimes \mathcal{E}_{1-p,d}^{\otimes k},
\end{equation}
for which the superadditivty effects observed in Eqs.~\eqref{eq:rocket-erasure-supadd}, and ~\eqref{eq:rocket-erasure-comp-supadd} can boost information transmission rates for each successive channel use as long as $k\leq n$. For $p=1/2$, these effects are identical for both the direct and complementary channels. The parameter $p$ in Eqs.~\eqref{eq:rocket-erasure-supadd},\eqref{eq:rocket-erasure-comp-supadd} can be adjusted to enhance the superadditivity of either the direct channel or its complement relative to the other. We employ these ideas to prove our main results below.

\begin{theorem}\label{th:main1}
Let $p=1/2$, and $n,\alpha\in \mathbb{N}$ be such that $n^{\alpha -2}>8$. Furthermore, let $d=2^{n^\alpha}$ and ${\cal N}= {\cal N}_{n,p,d}$. Then, for all $k\leq n:$
\begin{align}
  \Q{k+1}({\cal N}) - {\cal P}^{(k)}_{\max}({\cal N}) &\geq \frac{n^{\alpha} - 4n(k+1)}{2k(k+1)}>0, \\
  \Q{k+1}({\cal N}^c) - {\cal P}^{(k)}_{\max}({\cal N}) &\geq \frac{n^{\alpha} - 4n(k+1)}{2k(k+1)}>0,
\end{align}
\end{theorem}

\begin{proof}
    Using Lemma~\ref{lemma:directsum}, and the fact that for any channel $\cal N$, $\Ccap{1} ({\cal N}\otimes {\cal E}_{p,d}) = \Ccap{1}({\cal N})+ \Ccap{1} ({\cal E}_{p,d})$ \cite[Lemma 2]{Elkouss2015private}, we get:
\begin{align}
    \mathcal{P}^{(k)}( \mathcal{N}) &= \frac{1}{k}\max_{0\leq l\leq k} \mathcal{P}^{(1)}(R_d^{\otimes nl} \otimes \mathcal{E}_{1/2,d}^{\otimes k-l})  \nonumber \\ 
    &\leq \max \begin{cases}
        2n \\
        \frac{2n}{k} + \frac{k-1}{k} \frac{1}{2}\log d \\
        0
    \end{cases} \nonumber \\ 
    &= \frac{2n}{k} + \frac{(k-1)n^{\alpha}}{2k}.
\end{align}
 Analogous reasoning yields 
 a bound for $\mathcal{P}^{(k)}( \mathcal{N}^c)$. Combining Lemma~\ref{lemma:directsum}  with Eq.~\eqref{eq:rocket-erasure-supadd} and the fact that $\mathcal{Q}^{(1)}(\mathcal{N}_1\otimes \mathcal{N}_2)\geq \mathcal{Q}^{(1)}(\mathcal{N}_1)+\mathcal{Q}^{(1)}(\mathcal{N}_2)$ yields a simple lower bound: 
\begin{equation*}
    \mathcal{Q}^{(k+1)}( \mathcal{N}) \geq \frac{\mathcal{Q}^{(1)}( R_d^{\otimes n} \otimes \mathcal{E}_{1/2,d}^{\otimes k})}{k+1} \geq \frac{kn^{\alpha}}{2(k+1)}
\end{equation*}
for $k\leq n$. Swapping Eq.~\eqref{eq:rocket-erasure-supadd} with Eq.~\eqref{eq:rocket-erasure-comp-supadd} shows that the same bound holds for $\mathcal{Q}^{(k+1)}( \mathcal{N}^c)$ too. Thus,
\begin{align*}
 \Q{k+1}({\cal N}) - {\cal P}^{(k)}_{\max}({\cal N}) &\geq 
 \frac{n^{\alpha} - 4n(k+1)}{2k(k+1)}>0,
\end{align*}
where the latter inequality holds because $n^{\alpha-2}>8$. Clearly, the same bounds hold for ${\cal N}^c$ as well.
\end{proof}

When $p < 1/2$, the superadditivity analysis of ${\cal N}_{n,p,d}$ becomes tedious. In Lemma~\ref{lemma:bounds}, we prove several capacity bounds which we use to show that the superadditivity
of the direct channel dominates that of its complement:
\begin{theorem}\label{th:main2}
Let $p\in [0,1]$, and $n,\alpha\in \mathbb{N}$ be such that $1/3 < p \leq 1/2 - 1/n^{\alpha-1}$ and $n^{\alpha -2}>12$. Furthermore, let $d=2^{n^\alpha}$ and ${\cal N}= {\cal N}_{n,p,d}$. Then, 
\begin{align}
    \mathcal{Q}^{(k+1)}( {\cal N}) - \mathcal{P}^{(k)}_{\max}( {\cal N}) &\geq \frac{n^{\alpha}(1-p) - 2n(k+1)}{k(k+1)} \nonumber  \\ 
    &=: f_{n,p,\alpha}(k) > 0, \label{eq:thmain1} \\
    \mathcal{Q}^{(k+1)}({\cal N}^c) - \mathcal{P}^{(k)}({\cal N}^c) &\geq \frac{n^{\alpha}p - 2n(k+1)}{k(k+1)} \nonumber \\
    &=: f^c_{n,p,a}(k) > 0, \label{eq:thmain2} 
\end{align}
where the first bound holds for $2\leq k\leq n$ and the second bound holds for $1\leq k\leq n$.
\end{theorem}
The full proof is located in Appendix~\ref{appen:main-proofs}.
\begin{figure}[H]
 \hspace{-0.4cm}\includegraphics[scale=0.48]{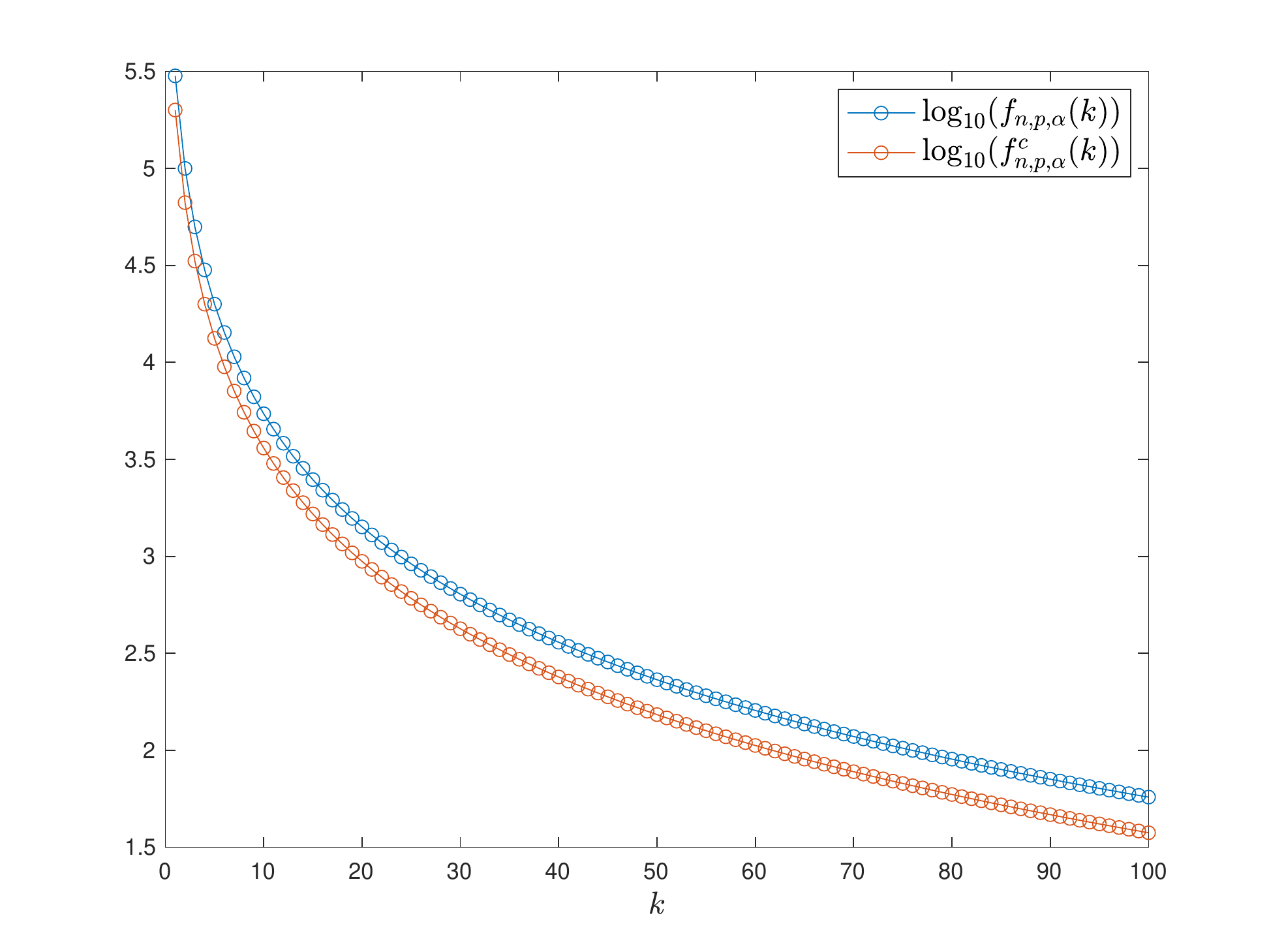}
    \caption{ \centering Plot of the log of the lower bounds $f_{n,p,\alpha}(k)$ and $f^c_{n,p,\alpha}(k)$ on the difference quantities $\mathcal{Q}^{(k+1)}( \mathcal{N}_{n,p,d}) - \mathcal{P}^{(k)}_{\max}( \mathcal{N}_{n,p,d})$ and $\mathcal{Q}^{(k+1)}( \mathcal{N}^c_{n,p,d}) - \mathcal{P}^{(k)}( \mathcal{N}^c_{n,p,d})$, respectively (see Theorem~\ref{th:main2}). Here, $n=100$, $p=0.4$, and $\alpha= 3$. }
    \label{fig:1}
\end{figure}

The log lower bounds for the difference quantities in Theorem~\ref{th:main2} are plotted in Figure~\ref{fig:1}. Note that in the setting of Theorem~\ref{th:main2}, the following is true:
\begin{align}
\frac{k-1}{k} &\geq \frac{2 +n^{\alpha}p}{(1-p)(n+1)n^{\alpha-1}} \nonumber \\ 
&\implies \Pcap{n+1}({\cal N}^c) \leq \Q{k}({\cal N}), \label{eq:supNvNc}
\end{align}
provided that $k\leq n$ (see Theorem~\ref{th:main222}). For instance, when $n=100, p=0.4,$ and $\alpha=3$, the LHS in Eq.~\eqref{eq:supNvNc} holds for $k\geq 3$. In other words, the direct channel in this case is vastly more superadditive than its complement, since only the 3-letter coherent information of ${\cal N}_{n,p,d}$ suffices to beat the 101-letter private information of ${\cal N}^c_{n,p,d}$.

For small values of $p$, the coherent information of ${\cal N}_{n,p,d}$ can even beat the total private information for some uses of the channel. For example, when $p=0.09, n=100,$ and $\alpha=3$, $\Q{101}({\cal N}_{n,p,d})$ beats $\mathcal{P}^{(k)}_{\rm{tot}}({\cal N}_{n,p,d})$ for $k\leq 9$ (see Figure~\ref{fig:2}). A detailed analysis of this phenomenon is given in Theorem~\ref{th:main3}.

\begin{figure}[H]
    \hspace{-0.4cm}\includegraphics[scale=0.48]{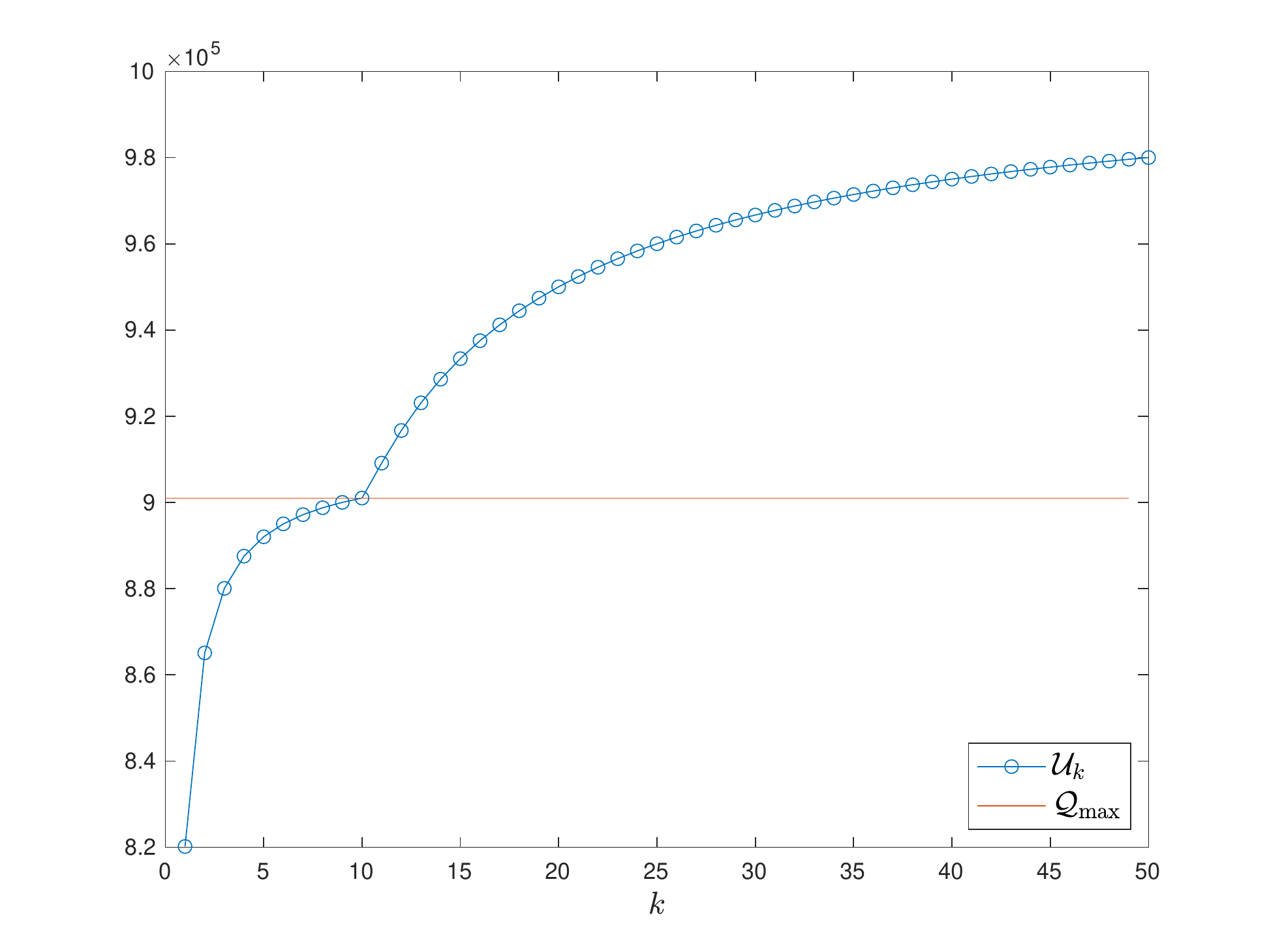}
    \caption{\centering Plot of the lower bound $Q_{max} = L(n+1)$ on  $\mathcal{Q}^{(n+1)}(\mathcal{N}_{n,p,d})$ and the upper bound $\mathcal{U}_k = \begin{cases}
     U'(k) \,\,\quad \text{if } k\leq k_0 \\
     U''(k) \quad \text{otherwise}
    \end{cases}$ on $\mathcal{P}^{(k)}_{\rm{tot}}(\mathcal{N}_{n,p,d})$ for $p=0.09, n=100$, and $\alpha=3$. The lower and upper bound functions are defined in Eqs.~\eqref{eq:Lk},\eqref{eq:U'k},\eqref{eq:U''k}. Clearly, $Q_{max}$ exceeds the total private capacity for at least $9$ uses of the channel.}
    \label{fig:2}
\end{figure}
\noindent
{\it {Platypus construction.--}} We now turn to constructing a channel with superadditive quantum capacity whose complement has additive capacity.
The $d-$dimensional platypus channel ${\cal M}_d: A\to B$ introduced recently in \cite{Siddhu2022platypus} is defined via the isometry: 
\begin{align}
    V:{\cal H}_A &\to {\cal H}_B \otimes {\cal H}_E \\
    V\ket{0} &= \frac{1}{\sqrt{d-1}} \sum_{j=0}^{d-2} \ket{j}\ket{j}, \nonumber \\
    V\ket{i} &= \ket{d-1}\ket{i-1}, \quad i=1,2,\ldots ,d-1, \nonumber
\end{align} 
as ${\cal M}_d(X) = \operatorname{Tr}_E (VXV^{\dagger})$. Here, $A$ and $B$ are $d-$dimensional while $E$ is $(d-1)-$dimensional. The channel and its complement satisfy \cite{Siddhu2022platypus, Siddhu2022platypus2}:
\begin{align}
   {\cal Q}({\cal M}_d) \leq \log \left( 1 + \frac{1}{\sqrt{d-1}} \right), \,\, {\cal P}({\cal M}_d) = {\cal C}({\cal M}_d) = 1,
\end{align}
\begin{align}
   \Q{1}({\cal M}^c_d) = {\cal Q}({\cal M}^c_d) = {\cal P}({\cal M}^c_d) &= {\cal C}({\cal M}^c_d) \nonumber \\
   &= \log (d-1).
\end{align}
As $d\to \infty, {\cal Q}({\cal M}_d)\to 0$. However, when coupled with an erasure channel, ${\cal M}_d$ can be shown to retain some quantum capacity as $d\to \infty$ \cite{Siddhu2022platypus}:
\begin{align}
    \Q{1}({\cal M}_{d+1} \otimes {\cal E}_{1/2,d}) \geq \frac{1}{2} + O\left( \frac{1}{\sqrt{d}} \right).
\end{align}

To turn these effects into superadditivity of a single channel, we again turn to a direct sum construction:
\begin{align}
    {\cal N}_d &:= {\cal M}_{d+1} \oplus {\cal E}_{1/2,d}, \\
    {\cal N}^c_d &:= {\cal M}^c_{d+1} \oplus {\cal E}_{1/2,d}.
\end{align}
By choosing a large enough $d$, we can ensure that $\Q{1}({\cal M}_{d+1} \otimes {\cal E}_{1/2,d}) > 2 \log (1 + 1/\sqrt{d}) \geq \Q{1}({\cal M}^{\otimes 2}_{d+1})$. In other words, ${\cal N}_d$ is superadditive for large enough $d$: 

\begin{align}
    \Q{2}({\cal N}_d) &= \frac{1}{2} \max \{\Q{1}({\cal M}^{\otimes 2}_{d+1}), \Q{1}({\cal M}_{d+1} \otimes {\cal E}_{1/2,d})\} \nonumber  \\ 
    &> \log \left( 1+ \frac{1}{\sqrt{d}} \right) \geq \Q{1}({\cal N}_d).
\end{align}

On the other hand, for any $k_1,k_2\in \mathbb{N}$:
\begin{align}
    \Q{1}({\cal M}^{c \, \otimes k_1}_{d+1} \otimes {\cal E}^{\otimes k_2}_{1/2,d}) &\leq \Ccap{1}({\cal M}^{c \, \otimes k_1}_{d+1} \otimes {\cal E}^{\otimes k_2}_{1/2,d}) \nonumber \\ 
    &= \Ccap{1}({\cal M}^{c \, \otimes k_1}_{d+1}) + \Ccap{1}({\cal E}^{\otimes k_2}_{1/2,d}) \nonumber \\
    &= k_1 \log d + \frac{k_2}{2}\log d \nonumber \\
    &\leq (k_1 + k_2) \log d.
\end{align}
Thus, ${\cal N}^c_d$ has additive quantum capacity:
\begin{align}
   \forall k: \,\,\, \Q{k}({\cal N}^c_d) &= \frac{1}{k} \max_{0\leq l\leq k} \Q{1} ({\cal M}^{c\, \otimes l}_{d+1} \otimes   \mathcal{E}_{1/2,d}^{\otimes k-l}) \nonumber \\
    &= \Q{k} ({\cal M}^c_{d+1}) = \log d.
\end{align}

\noindent
{\it {Discussion.--}}
We have investigated superadditive effects for the coherent and private information of a channel and its complement relative to each other. We showed that contrary to intuitive expectations, the following two cases are both possible:
\begin{itemize}
    \item The direct and complementary channels are simultaneously superadditive,
    \item The direct channel is superadditive while the complement is additive (and vice versa).
\end{itemize}
It is also possible to construct examples where both the direct and complementary channels are entanglement-breaking and hence trivially have additive coherent and private information (equal to zero) \cite{Singh2022bippt}.

One interesting question to investigate further is the extent to which superadditivity of the coherent and private information quantities can be attained simultaneously for the direct and complementary channels. Is there a constraint akin to the monogamy of entanglement which would prohibit a maximum violation of additivity for both channels? Another intriguing question is whether one can superactivate total capacity.  

Finally, it would be enlightening to check whether or not the superadditivity results presented in this letter hold for the classical capacities as well. Our techniques do not apply directly in this case. The problem lies in the fact that the classical capacity of a direct sum channel is not merely the maximum of the classical capacity of its components. This is because when using a direct sum channel, say $\mathcal{N}= \oplus_{i=1}^n \mathcal{N}^{(i)}$, unlike private or coherent information, classical information can not only be sent through the individual blocks $\mathcal{N}^{(i)}$, but can also be encoded in the choice of the blocks $i=1,2,\ldots ,n$. More precisely, for our channel of interest, say ${\cal N} = R_d \oplus {\cal E}_{p,d}$,
Lemma~\ref{lemma:directsum} shows that 
\begin{align*}
    \Ccap{1}(\mathcal{N}) &= \log( 2^{\Ccap{1}(R_d)} + 2^{\Ccap{1}(\mathcal{E}_{p,d})}), \,\, \rm{and} \\
    \Ccap{2}(\mathcal{N}) &=  \frac{1}{2}\log( 2^{\Ccap{1}(R_d^{\otimes 2})} + 2.2^{\Ccap{1}(R_d\otimes \mathcal{E}_{p,d})} + 2^{\Ccap{1}(\mathcal{E}^{\otimes 2}_{p,d})}).
\end{align*}
Since $\Ccap{1}(\mathcal{N}\otimes \mathcal{E}_{p,d}) = \Ccap{1}(\mathcal{N})+\Ccap{1}(\mathcal{E}_{p,d})$ for all channels $\mathcal{N}$ \cite[Lemma 2]{Elkouss2015private}, it is clear that the question of superadditivity of  $\Ccap{1}(\mathcal{N})$ boils down to the question of superadditivity of $\Ccap{1}(R_d)$, which is currently open. 

{\it Acknowledgements.}
The authors would like to thank Nilanjana Datta for helpful discussions. 
Satvik Singh acknowledges support from the Cambridge Trust's International Scholarship.
Sergii Strelchuk acknowledges support from the Royal Society University Research Fellowship.  
\bigskip

\onecolumngrid

\PRLsep

\bibliography{references}


\appendix

\counterwithin{theorem}{section}

\section{Notation and the direct sum construction}

We denote quantum systems with capital letters like $A$ and $B$. Each quantum system $A$ is associated with a finite-dimensional complex Hilbert space $\mathcal{H}_A$. For a joint system $AB$ or $A\oplus B$, $\mathcal{H}_{AB}=\mathcal{H}_A\otimes \mathcal{H}_B$ or ${\cal H}_{A\oplus B} = {\cal H}_A \oplus {\cal H}_B$. Quantum states $\rho_A$ on $A$ are positive semi-definite operators acting on $\mathcal{H}_A$ with unit trace. The von Neumann entropy of $\rho_A$ is defined as $S(A) := -\operatorname{Tr}(\rho_A \log_2 \rho_A)$. For a bipartite state $\rho_{AB}$,the mutual information is defined as $I(A:B):= S(A)+S(B)-S(AB)$.

A quantum channel $\mathcal{N}:A\to B$ (denoted $\mathcal{N}_{A\to B}$) is a completely positive and trace-preserving linear map taking linear operators acting on $\mathcal{H}_A$ to those acting on $\mathcal{H}_B$. Every channel ${\cal N}_{A\to B}$ admits a Stinespring isometry $V:{\cal H}_A \to {\cal H}_B \otimes {\cal H}_E$ such that $\mathcal{N}(X)=\operatorname{Tr}_E(VX V^\dagger)$. The complementary channel ${\cal N}^c_{A\to E}$ is defined as $\mathcal{N}^c (X)  = \operatorname{Tr}_B(VXV^\dagger)$. For two channels ${\cal N}^{(1)}:A_1\to B_1, {\cal N}^{(2)}:A_2\to B_2$, the direct sum ${\cal N}= {\cal N}^{(1)} \oplus {\cal N}^{(2)} : A_1\oplus A_2 \to B_1 \oplus B_2$ acts as 
\begin{align}
    {\cal N}\left(
    \begin{bmatrix} X_{A_1} & * \\ * & X_{A_2} \end{bmatrix} \right) = \begin{bmatrix} {\cal N}^{(1)}(X_{A_1}) & 0 \\ 0 & {\cal N}^{(2)}(X_{A_2}) \end{bmatrix}.
\end{align}
It should be clear that the complement ${\cal N}_c = {\cal N}^{(1)}_c \oplus {\cal N}^{(2)}_c$.

\begin{lemma}\label{lemma:directsum}
For quantum channels $\mathcal{N}^{(i)}_{A_i \to B_i}$ for $i=1,2,\ldots n$, let $\mathcal{N} := \oplus_{i=1}^n \mathcal{N}^{(i)}$. Then,
\begin{align}
    \text{\cite[\rm{\, Proposition \, 1}]{Fukuda2007additivity}} \qquad \mathcal{Q}^{(1)}(\mathcal{N}) &= \max_{1\leq i\leq n} \mathcal{Q}^{(1)}(\mathcal{N}^{(i)}), \\
    \text{\cite[\rm{\, Lemma \, 1}]{Elkouss2015private}} \qquad \mathcal{P}^{(1)}(\mathcal{N}) &= \max_{1\leq i\leq n} \mathcal{P}^{(1)}(\mathcal{N}^{(i)}), \\
    \text{\cite[\rm{\, Proposition \, 1}]{Fukuda2007additivity}} \qquad \Ccap{1}(\mathcal{N}) &= \log \sum_{i=1}^n 2^{ \Ccap{1}(\mathcal{N}^{(i)})}.
\end{align}
\end{lemma}

\section{Proofs} \label{appen:main-proofs}

Recall from the main text that for $n,d\in \mathbb{N}$ and $p\in [0,1]$, the channel ${\cal N}_{n,p,d}$ was defined as the direct sum ${\cal N}_{n,p,d}:= R^{\otimes n}_{d} \oplus \mathcal{E}_{p,d}$, where $R_d$ and ${\cal E}_{p,d}$ are the rocket and erasure channels. We can derive the following bounds on the capacities of this channel.

\begin{lemma} \label{lemma:bounds}
Let $p\in [0,1]$ and $n,\alpha\in \mathbb{N} \, (\alpha > 1)$ be such that $4/n^{\alpha -1} \leq p \leq 1/2 - 1/n^{\alpha-1}$. Let $d = 2^{n^\alpha}$ and 
\begin{equation*}
    k_0 (n,p,\alpha) := \frac{1-p-2/n^{\alpha -1}}{p}.
\end{equation*}
Then, the following bounds hold:
\begin{alignat}{2}
    \forall k\leq k_0 &: \quad \mathcal{P}^{(k)}( \mathcal{N}_{n,p,d} )  \leq (1-2p)n^{\alpha}\label{eq:bound1} \\
    \forall k> k_0 &: \quad \mathcal{P}^{(k)}( \mathcal{N}_{n,p,d} ) \leq \frac{2n}{k} + \frac{k-1}{k} (1-p)n^{\alpha} \label{eq:bound2} \\
    \forall k &: \quad \mathcal{P}^{(k)}( \mathcal{N}_{n,p,d}^c)  \leq \frac{2n}{k} + \frac{k-1}{k} pn^{\alpha} \label{eq:bound3} \\
    \forall k\leq n &: \quad \mathcal{Q}^{(k+1)}( \mathcal{N}_{n,p,d} ) \geq \frac{k}{k+1}(1-p)n^{\alpha} \label{eq:bound4} \\
    \forall k\leq n &: \quad \mathcal{Q}^{(k+1)}( \mathcal{N}^c_{n,p,d} ) \geq \frac{k}{k+1}pn^{\alpha} \label{eq:bound5}
\end{alignat}
\end{lemma}

\begin{proof}
By using Lemma~\ref{lemma:directsum} and the fact that $\Ccap{1} ({\cal N}\otimes {\cal E}_{p,d}) = \Ccap{1}({\cal N})+ \Ccap{1} ({\cal E}_{p,d})$ for all channels $\cal N$, it is easy to arrive at the following bounds: 
\begin{align}
    \mathcal{P}^{(k)}( \mathcal{N}_{n,p,d}) &= \frac{1}{k}\max_{0\leq l\leq k} \mathcal{P}^{(1)}(R_d^{\otimes nl} \otimes \mathcal{E}_{p,d}^{\otimes k-l}) \nonumber \\ 
    &\leq \max \begin{cases}
        2n \\
        \frac{2n}{k} + \frac{k-1}{k} (1-p) n^{\alpha} \\
        (1-2p) n^{\alpha},
    \end{cases} \label{eq:P1k}  \\
    \mathcal{P}^{(k)}( \mathcal{N}_{n,p,d}^c) &= \frac{1}{k} \max_{0\leq l\leq k} \mathcal{P}^{(1)}(R_d^{c \,\otimes nl} \otimes \mathcal{E}_{1-p,d}^{\otimes k-l}) \nonumber \\ 
    &\leq \max 
    \begin{cases}
     2n \\ 
     \frac{2n}{k} + \frac{k-1}{k} p n^{\alpha}.
    \end{cases}
\end{align}

Let's first deal with the maximum in Eq.~\eqref{eq:P1k}. Clearly, $(1-2p)n^{\alpha} \geq 2n$ since  $p \leq 1/2 - 1/n^{\alpha -1}$. Moreover, 
\begin{align*}
    &\frac{2n}{k} + \frac{k-1}{k} (1-p)n^{\alpha} \leq  (1-2p)n^{\alpha} \\
    &\iff \frac{2}{k} + (1-\frac{1}{k}) (1-p)n^{\alpha -1} \leq (1-2p)n^{\alpha - 1} \\ 
    & \iff \frac{1}{k}((1-p)n^{\alpha - 1} - 2) \geq n^{\alpha - 1}p \\ 
    & \iff k \leq \frac{(1-p - 2/n^{\alpha -1})}{p} = k_0.
\end{align*}
This gives us the first two bounds in Eqs. \eqref{eq:bound1} and \eqref{eq:bound2}. To obtain the bound in Eq.~\eqref{eq:bound3}, observe that
\begin{equation}
    4/n^{\alpha -1} \leq p \iff 2n \leq \frac{1}{2} pn^{\alpha}.
\end{equation}
Hence, for $k\geq 2$, we have $ \frac{k-1}{k} pn^{\alpha} \geq pn^{\alpha}/2 \geq 2n$. To prove Eq.~\eqref{eq:bound4}, note that for $k\leq n$, Lemma~\ref{lemma:directsum} shows
\begin{align*}
    \mathcal{Q}^{(k+1)}( \mathcal{N}_{n,p,d}) \geq \frac{\mathcal{Q}^{(1)}( R_d^{\otimes n} \otimes \mathcal{E}_{p,d}^{\otimes k})}{k+1} \geq \frac{k}{k+1} (1-p)n^{\alpha},
\end{align*}
where we have used the bound in Eq.~\eqref{eq:rocket-erasure-supadd} along with the fact that $\mathcal{Q}^{(1)}(\mathcal{N}_1\otimes \mathcal{N}_2)\geq \mathcal{Q}^{(1)}(\mathcal{N}_1)+\mathcal{Q}^{(1)}(\mathcal{N}_2)$. An identical argument works to prove Eq.~\eqref{eq:bound5}.
\end{proof}

We can now prove Theorem~\ref{th:main2} from the main text.

\begin{theorem}\label{th:main22}
Let $p\in [0,1]$, and $n,\alpha\in \mathbb{N}$ be such that $1/3 < p \leq 1/2 - 1/n^{\alpha-1}$ and $n^{\alpha -2}>12$. Furthermore, let $d=2^{n^\alpha}$ and ${\cal N}= {\cal N}_{n,p,d}$. Then, 
\begin{align}
    \mathcal{Q}^{(k+1)}( {\cal N}) - \mathcal{P}^{(k)}_{\max}( {\cal N}) &\geq \frac{n^{\alpha}(1-p) - 2n(k+1)}{k(k+1)} \nonumber  \\ 
    &=: f_{n,p,\alpha}(k) > 0, \label{eq:thmain11} \\
    \mathcal{Q}^{(k+1)}({\cal N}^c) - \mathcal{P}^{(k)}({\cal N}^c) &\geq \frac{n^{\alpha }p - 2n(k+1)}{k(k+1)} \nonumber \\
    &=: f^c_{n,p,a}(k) > 0, \label{eq:thmain22} 
\end{align}
where the first bound holds for $2\leq k\leq n$ and the second bound holds for $1\leq k\leq n$.  $\Q{2}({\cal N})>\Pcap{1}_{\max}({\cal N})$ 
\end{theorem}

\begin{proof}
For $1/3 < p \leq 1/2 - 1/n^{\alpha -1}$, Lemma~\ref{lemma:bounds} shows that 
\begin{equation} \label{eq:Uk}
    \mathcal{P}^{(k)}_{\max}( \mathcal{N}_{n,p,d})\leq \frac{2n}{k} + \frac{k-1}{k} (1-p)n^{\alpha} =: U_{n,p,\alpha}(k)
\end{equation}
for $2\leq k\leq n$. Furthermore, for $k\leq n$, we have
\begin{equation}\label{eq:Lk}
    \mathcal{Q}^{(k+1)}( \mathcal{N}_{n,p,d}) \geq \frac{k}{k+1}(1-p)n^{\alpha} =: L_{n,p,\alpha}(k+1) .
\end{equation}
Now,
\begin{align*}
    L(k+1) - U(k) &= \frac{k}{k+1} (1-p)n^{\alpha} - \frac{2n}{k} - \frac{k-1}{k} (1-p)n^{\alpha} \\
    &= \frac{n^{\alpha}(1-p) - 2n(k+1)}{k(k+1)} > 0,
\end{align*}
where the final inequality holds since $n^{\alpha -2}>12>8$, $k\leq n$, and $1-p\geq 1/2$:
\begin{equation*}
    2(k+1) \leq 2(n+1) \leq 4n < \frac{n^{\alpha -1}}{2} \leq n^{\alpha -1} (1-p).
\end{equation*}
This establishes Eq.~\eqref{eq:thmain11}. To prove Eq.~\eqref{eq:thmain22}, we again use Lemma~\ref{lemma:bounds} to obtain (for $k\leq n):$
\begin{align}
    \mathcal{P}^{(k)}( \mathcal{N}^c_{n,p,d}) &\leq \frac{2n}{k} + \frac{k-1}{k} pn^{\alpha} =: U^c_{n,p,\alpha}(k), \label{eq:Ukc} \\
    \mathcal{Q}^{(k+1)}( \mathcal{N}^c_{n,p,d}) &\geq \frac{k}{k+1}pn^{\alpha} =: L^c_{n,p,\alpha}(k+1).  \label{eq:Lkc}
\end{align}
As before, the claim follows by noting that 
\begin{align*}
    L^c(k+1) - U^c(k) &= \frac{k}{k+1} pn^{\alpha} - \frac{2n}{k} - \frac{k-1}{k} pn^{\alpha} \\
    &= \frac{n^{\alpha}p - 2n(k+1)}{k(k+1)} > 0,
\end{align*}
where the final inequality is true because $n^{\alpha -2}>12$, $k\leq n$, and $p> 1/3$:
\begin{equation*}
    2(k+1) \leq 2(n+1) \leq 4n < \frac{n^{\alpha -1}}{3} < n^{\alpha -1} p.
\end{equation*}
\end{proof}

\begin{theorem}   \label{th:main222} 
In the setting of Theorem~\ref{th:main2}, for $k\leq n$, the following implication holds:
\begin{align}
\frac{k-1}{k} &\geq \frac{2 +n^{\alpha}p}{(1-p)(n+1)n^{\alpha-1}} \nonumber \\ 
&\implies \Pcap{n+1}({\cal N}^c) \leq \Q{k}({\cal N}). 
\end{align}
\end{theorem}
\begin{proof}
    For $k\leq n$, we have 
\begin{align*}
    \Pcap{k+1}({\cal N}^c_{n,p,d}) &\leq U^c_{n,p,\alpha}(k+1), \\
    \Q{k}({\cal N}^c_{n,p,d}) &\geq L_{n,p,\alpha}(k),
\end{align*}
where $U^c$ and $L$ are defined in Eq.~\eqref{eq:Ukc}, \eqref{eq:Lk}. The claim follows by noting that $L(k)\geq U^c(n+1)$ if and only if $k$ satisfies the desired inequality.
\end{proof}

\begin{theorem}\label{th:main3}
Let $p\in [0,1]$ and $n,\alpha\in \mathbb{N} \, (\alpha > 1)$ be such that $4/n^{\alpha -1} \leq p \leq 1/2 - 1/n^{\alpha-1}$. Let $d = 2^{n^\alpha}$ and $k_0 := (1-p-2/n^{\alpha -1})/p$. Fix $k\leq k_0$ and define 
\begin{equation}
    c_{n,p,\alpha}(k) := \frac{(1-p)k}{p - 2/n^{\alpha-1}}.
\end{equation}
Then, for $c(k) < j \leq n+1$ (provided such a $j$ exists),
\begin{equation}
    \mathcal{Q}^{(j)}( \mathcal{N}_{n,p,d} ) > \mathcal{P}^{(k)}_{\rm{tot}}( \mathcal{N}_{n,p,d} ) .
\end{equation}
\end{theorem}

\begin{proof}
Lemma~\ref{lemma:bounds} shows that for $k\leq k_0:$
\begin{align} \label{eq:U'k}
     \mathcal{P}^{(k)}_{\rm{tot}}( \mathcal{N}_{n,p,d} ) \leq (1-2p)n^{\alpha} + \frac{2n}{k} + \frac{k-1}{k} pn^{\alpha} =: U'_{n,p,\alpha}(k), 
\end{align}
and $\mathcal{Q}^{(j)}( \mathcal{N}_{n,p,d}) \geq \frac{j-1}{j} (1-p)n^{\alpha} =: L_{n,p,\alpha}(j)$ for $j\leq n+1$. The result follows by noting that
\begin{align*}
    L_{n,p,\alpha}(j) > U'_{n,p,\alpha}(k) \iff j > \frac{(1-p)k}{p-2/n^{\alpha-1}}. 
\end{align*}
\end{proof}

A similar reasoning as above can be applied when $k>k_0$. In this case, Lemma~\ref{lemma:bounds} shows that 
\begin{equation}\label{eq:U''k}
    \forall k>k_0: \quad \mathcal{P}^{(k)}_{\rm{tot}}( \mathcal{N}_{n,p,d})\leq \frac{4n}{k} + \frac{k-1}{k} n^{\alpha} =: U''_{n,\alpha}(k),
\end{equation}
and the lower bound for $\mathcal{Q}^{(j)}$ remains the same: $\mathcal{Q}^{(j)}( \mathcal{N}_{n,p,d}) \geq L_{n,p,\alpha}(j)$ for $j\leq n+1$. Thus, we have
\begin{equation*}
    L_{n,p,\alpha}(j) > U''_{n,\alpha}(k) \iff \frac{j-1}{j} > \frac{4/n^{\alpha -1}+ k-1}{k(1-p)},
\end{equation*}
provided such a $j\leq n+1$ exists.
\newpage

\section{Rocket channels}\label{appen:rocket}

In this section, we provide a graphical proof of the superadditive nature of the Rocket channel when used jointly with erasure (Eq.~\eqref{eq:rocket-erasure-supadd}, \eqref{eq:rocket-erasure-comp-supadd}). The graphical presentation makes for a more intuitive and easily digestible argument. A quick summary of the necessary diagrammatic notation can be found in \cite[Section 3]{nechita2021graphical}. For more detailed expositions, we refer the readers to \cite{wood2015tensor, bridgeman2017hand}.

Recall that the rocket channel $R_d: A_1 \otimes A_2 \to B$ first applies local (independent) random unitaries $U,V$ on inputs $A_1, A_2$ respectively, and then couples them via a controlled phase gate $P=\sum_{i,j}\omega^{ij}\ketbra{i}_{A_1}\otimes \ketbra{j}_{A_2}$, where $\omega=e^{i2\pi /d}$. Finally, $A_2$ is discarded and Bob gets $A_1$ along with classical information about which local unitaries were applied. Hence, each random instance of the Rocket channel acting on an input $X_{A_1 A_2}$ can be depicted as follows:

\begin{equation*}
   \includegraphics[align=c]{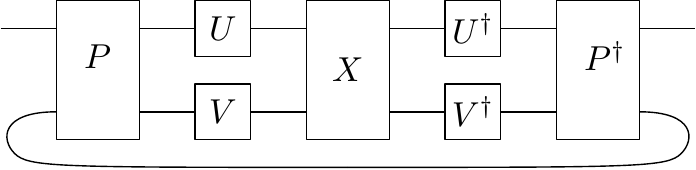},
\end{equation*}
where we have not shown the classical knowledge about $U,V$ that is also delivered to Bob. Note that the final output state will be obtained by taking an expectation over the random variables $U,V$. The complementary channel $R^c_d: A_1 \otimes A_2 \to E$ acts nearly identically, except that in the final step, $A_1$ is discarded and Eve gets $A_2$ along with the same classical information about which local unitaries were applied (which is not depicted below):

\begin{equation*}
   \includegraphics[align=c]{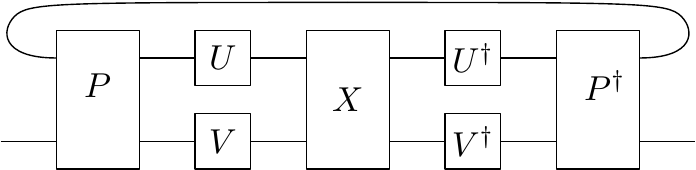}.
\end{equation*}
To transmit information by jointly using the Rocket channel (or its complement) with erasure, Alice first prepares a maximally entangled state and sends one half of it through the erasure channel ${\cal E}_{p,d}$. This establishes maximal entanglement between her and the receiver with probability $1-p$. Using this shared entanglement (shown in red in Figures \ref{fig:Rd} and \ref{fig:Rdc}), Alice and Bob can communicate through $R_d$ and $R^c_d$ at rate $\log d$ by following the steps described in Figures \ref{fig:Rd} and \ref{fig:Rdc}, respectively. With probability $p$, the erasure channel fails to distribute entanglement between Alice and the receiver, in which case the protocol fails. Hence, we get the following net rates of quantum communication:
\begin{equation*}
    \mathcal{Q}^{(1)}(R_d \otimes \mathcal{E}_{p,d}) \geq (1-p)\log d \quad\text{and}\quad \mathcal{Q}^{(1)}(R^c_d \otimes \mathcal{E}_{p,d}) \geq (1-p)\log d.
\end{equation*}

\begin{figure}[H]
    \centering
   \hspace{0.1cm}\includegraphics[scale=0.95]{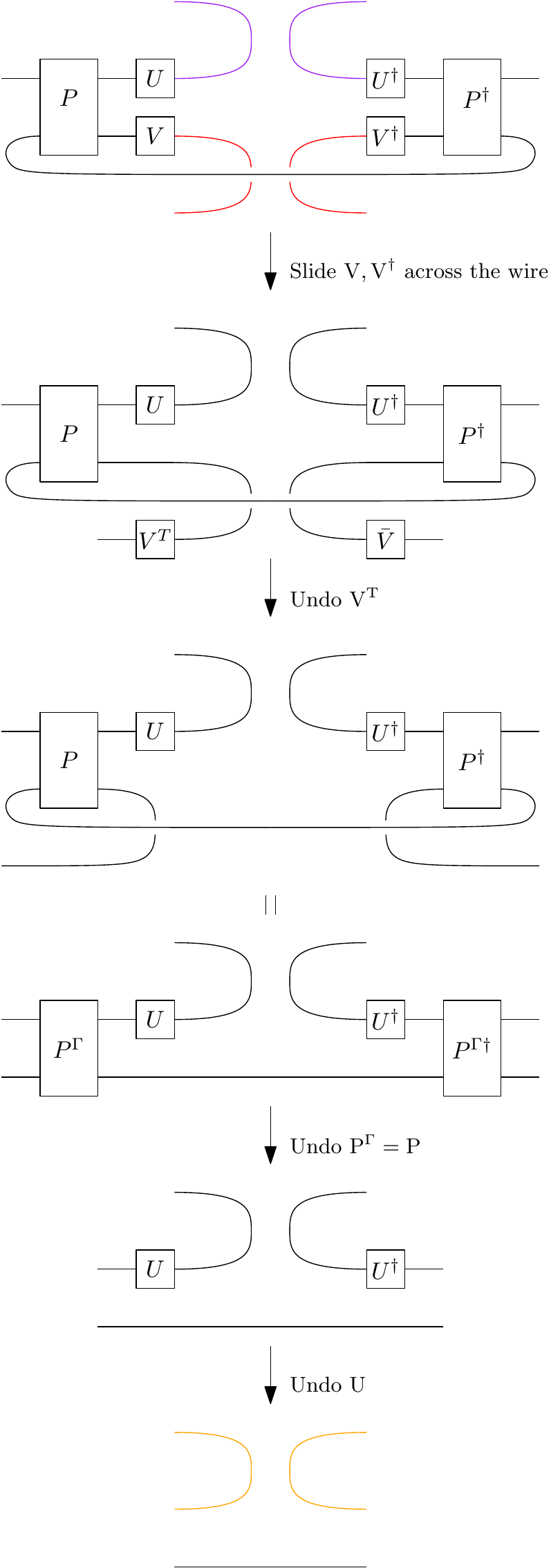}
   \vspace{0.5cm}
    \caption{Visual depiction of how the Rocket channel $R_d$ can be used to transmit information with the help of pre-shared entanglement. Alice and Bob start with a pre-shared maximally entangled state (shown in red). Alice locally prepares another maximally entangled state (shown in purple) and sends half of each of the entangled states through $R_d$ to Bob as shown. Hence, only the top two dangling wires are in Alice's possession while the bottom four are with Bob. Since Bob knows which local random unitaries $U,V$ are applied during the transmission, he can use the pre-shared entanglement with Alice to undo the phase coupling operation as described. Here, $P^{\Gamma}$ is the partial transpose of $P$ with respect to the second subsystem. The two parties finally end up sharing one maximally entangled state (shown in orange), which can be used to send quantum information at rate $\log d$.}
    \label{fig:Rd}
\end{figure}

\begin{figure}[H]
    \centering
    \hspace{0.1cm}\includegraphics[scale=0.95]{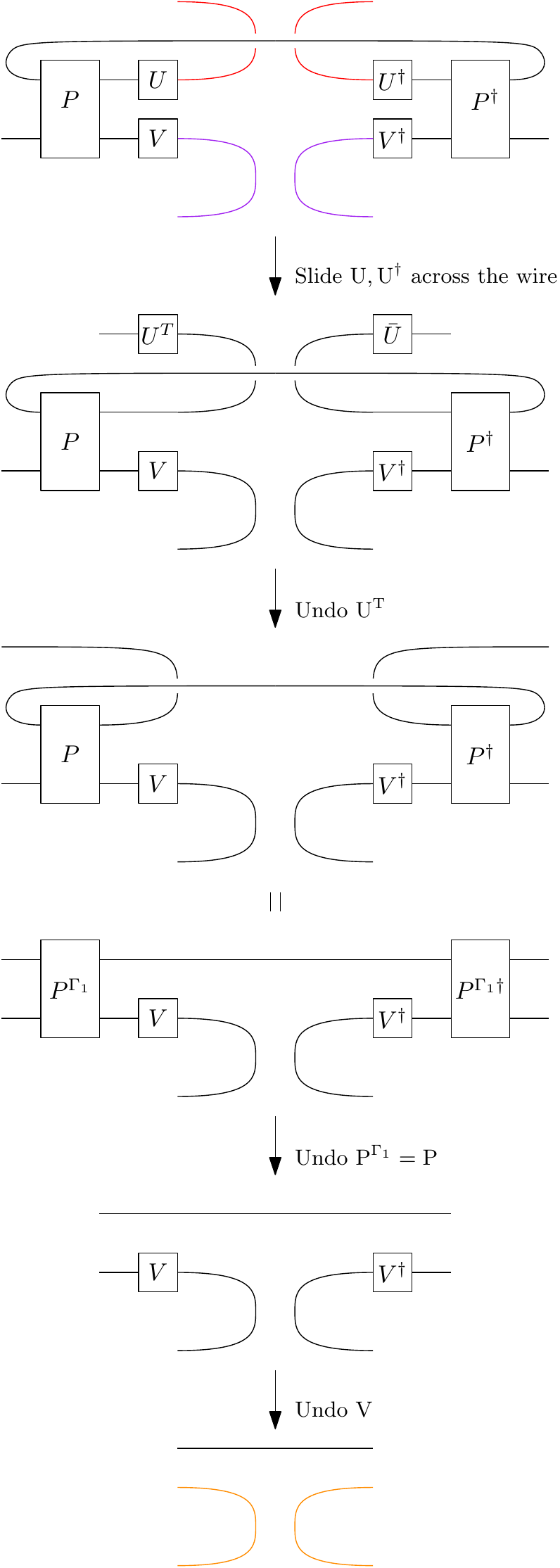}
    \vspace{0.5cm}
    \caption{Visual depiction of how the complementary Rocket channel $R^c_d$ can be used to transmit information with the help of pre-shared entanglement. Alice and Bob start with a pre-shared maximally entangled state (shown in red). Alice locally prepares another maximally entangled state (shown in purple) and sends half of each of the entangled states through $R^c_d$ to Bob as shown. Hence, only the bottom two dangling wires are in Alice's possession while the top four are with Bob. Since Bob knows which local random unitaries $U,V$ are applied during the transmission, he can use the pre-shared entanglement with Alice to undo the phase coupling operation as described. Here, $P^{\Gamma_1}$ is the partial transpose of $P$ with respect to the first subsystem. The two parties finally end up sharing one maximally entangled state (shown in orange), which can be used to send quantum information at rate $\log d$.}
    \label{fig:Rdc}
\end{figure}

\end{document}